\documentclass[conference,letterpaper]{IEEEtran}

\usepackage[utf8]{inputenc} 
\usepackage[T1]{fontenc}
\usepackage{url}
\usepackage{ifthen}
\usepackage{cite}
\usepackage[cmex10]{amsmath} 
\usepackage{amsmath}
\usepackage{amsthm}
\usepackage{amsfonts}
\usepackage{amssymb}
\usepackage{mathtools}
\usepackage{color}
\usepackage{verbatim}
\usepackage{mathtools}
\usepackage[normalem]{ulem}
\usepackage{enumitem}
\usepackage{hyperref}
\usepackage[ruled,vlined]{algorithm2e}
\usepackage{algpseudocode}

\definecolor{purple}{rgb}{0.5, 0.0, 0.5}
\definecolor{dark_green}{rgb}{0.0, 0.5, 0.0}

\newcommand{\dgreen}{\color{dark_green}}
\newcommand{\blue}{\color{blue}}
\newcommand{\cyan}{\color{cyan}}
\newcommand{\purple}{\color{purple}}

\newcommand{\I}{\mathcal{I}}
\newcommand{\J}{\mathcal{J}}

\newcommand{\xb}{\bold{x}}

\newcommand{\D}{\mathcal{D}}
\newcommand{\T}{\mathcal{T}}

\newcommand{\Z}{\mathbb{Z}}
\newcommand{\R}{\mathbb{R}}
\newcommand{\C}{\mathbb{C}}

\newcommand{\E}{\mathbb{E}}
\newcommand{\N}{\mathbb{N}}
\newcommand{\Q}{\mathbb{Q}}
\newcommand{\gb}{\bold{g}}

\newcommand{\Ab}{\bold{A}}
\newcommand{\Bb}{\bold{B}}
\newcommand{\ab}{\bold{a}}

\newcommand{\eb}{\bold{e}}
\newcommand{\rt}{\tilde{r}}
\newcommand{\Cfr}{\mathfrak{C}}

\DeclarePairedDelimiter\floor{\lfloor}{\rfloor}

\newtheorem{Thm}{Theorem}

\newtheorem{Prop}[Thm]{Proposition}
\newtheorem{Lemma}[Thm]{Lemma}
\newtheorem{Def}[Thm]{Definition}

\newcommand{\ind}{\text{\color{white}.$\quad$}}

\begin{document}
\title{Numerically Stable Binary Gradient Coding}

\author{
  \IEEEauthorblockN{{\bf Neophytos Charalambides, Hessam Mahdavifar, and Alfred O. Hero III}}
  \IEEEauthorblockA{Department of Electrical Engineering and Computer Science, University of Michigan, Ann Arbor, MI 48109\\
  Email: neochara, hessam, hero@umich.edu}
\vspace{-4mm}
}

\maketitle

\begin{abstract}
A major hurdle in machine learning is scalability to massive datasets. One approach to overcoming this is to distribute the computational tasks among several workers. \textit{Gradient coding} has been recently proposed in distributed optimization to compute the gradient of an objective function using multiple, possibly unreliable, worker nodes. By designing distributed coded schemes, gradient coded computations can be made resilient to \textit{stragglers}, nodes with longer response time compared to other nodes in a distributed network. Most such schemes rely on operations over the real or complex numbers and are inherently numerically unstable. We present a binary scheme which avoids such operations, thereby enabling numerically stable distributed computation of the gradient. Also, some restricting assumptions in prior work are dropped, and a more efficient decoding is given.
\end{abstract}


\section{Introduction}
\label{intro}

In modern day machine learning the \textit{curse of dimensionality} has been a major impediment to solving large scale problems, which require heavy computations. Recently, coding-theoretic ideas have been adopted in order to accommodate such computational tasks in a distributed manner, under the assumption that straggler workers are present \cite{lee2018speeding,li2016unified,reisizadeh2017coded,li2016coded,li2017coding,yang2017computing,lee2017high,dutta2016short,vulimiri2013low,wang2018coded,mallick2018rateless,ramamoorthy2019universally,YSRKSA18,RRG20}. Stragglers are workers whose tasks may never be received, due to delay or outage, and can significantly increase the computation time. These failures translate to erasures in the context of coding theory when the computational tasks are encoded. The authors of \cite{TLDK17} proposed \textit{gradient coding}, a scheme for exact recovery of the gradient when the objective loss function is additively separable. Gradient coding requires the central server to receive the subtasks of a fixed fraction of any of the workers. The exact recovery of the gradient is considered in several prior works, e.g., \cite{TLDK17,HASH17,RTTD17,OGU19}, while the numerical stability issue is studied in \cite{YA18}. Numerical stability for matrix multiplication is considered in \cite{FC19,SHN19}. There are also several works involving gradient coding for approximate recovery of the gradient \cite{CP18,RTTD17,CPE17,WCP19,BWE19,WLS19,KKR19,HYKM19,CHZP18,CPH20a}.

In this paper, we propose a scheme for gradient coding that is numerically stable. More specifically, the proposed scheme avoids any division or multiplication of real or complex numbers, often represented by floating point. Furthermore, the encoding matrix is binary, simplifying the encoding process. Also, the scheme is deterministic, i.e., it does not require generating random numbers. Our scheme is similar in spirit to the \textit{fractional repetition scheme} introduced in \cite{ERK10,TLDK17}, where we also drop the strict assumption that $s+1$ divides $n$, where $n$ is the number of workers and $s$ is the number of stragglers that the scheme tolerates. Additionally, the scheme we propose can also be adapted to compute matrix-matrix multiplications, and matrix inverse approximations \cite{CPH20b}. 

The main advantage of considering encoding and decoding real-valued data using binary matrices, consisting of $0$'s and $1$'s, is that it does not introduce further instability, possibly adding to the computational instability of the associated computation tasks, which was also considered for matrix-vector multiplication in \cite{JSM19}. The fact that the encoding matrix is over $\{0,1\}$ allows us to view the encoding as task assignments. This also gives a more efficient online decoding, which avoids searching through an exponentially large table in terms of $n$, as in the scheme in \cite{TLDK17}.

The paper is organized as follows. In section \ref{str_problem_GC} we overview the ``straggler problem'' in gradient coding \cite{TLDK17}. In sections \ref{bin_gc_sec} and \ref{dec_vector_sec} the proposed encoding and decoding processes are discussed, respectively. In section \ref{cl_un_ass_distr_sec} we discuss the optimality of our scheme. Finally, in section \ref{heter_case_sec} we consider scenarios with heterogeneous workers. The main contributions are:
\begin{itemize}[noitemsep,nolistsep]
  \item A new \textit{binary} gradient coding scheme --- both in the encoding and decoding, that is robust to stragglers;
  \item Elimination of the assumption $(s+1)\mid n$;
  \item Theoretically showing that perfect gradient recovery occurs, and that the unbalanced assignment is optimal assuming the homogeneous worker setting;
  \item Comparison with another binary scheme \cite{TLDK17}, justifying ours is more efficient, for large $n$;
  \item Determining task assignments for heterogeneous workers.
\end{itemize}

\section{Preliminaries}
\label{str_problem_GC}

\subsection{Straggler Problem}
\label{str_problem}

Consider a single central server that has at its disposal a dataset $\D=\left\{(\xb_i,y_i)\right\}_{i=1}^N\subsetneq \R^p\times\R$ of $N$ samples, where $\xb_i$'s represent the features and $y_i$ denotes the label of the $i$-th sample. The central server distributes the dataset $\D$ among $n$ workers to facilitate the solution of the optimization problem:
\begin{equation}
\label{th_star_pr}
  \theta^{\star} = \arg\min_{\theta\in\R^p}\left\{ \sum_{i=1}^N \ell(\xb_i,y_i;\theta) \right\}
\end{equation}
in an accelerated manner, where $L(\D;\theta)=\sum_{i=1}^N \ell(\xb_i,y_i;\theta)$ is a predetermined differentiable loss-function. The objective function in (\ref{th_star_pr}) can also include a regularizer $\mu R(\theta)$ if necessary. A common approach to solving (\ref{th_star_pr}) is to employ gradient descent. Even if closed form solutions exist for (\ref{th_star_pr}), gradient descent can still be advantageous for large $N$.

The central server is assumed to be capable of distributing the dataset appropriately, with a certain level of redundancy, in order to recover the gradient based on the full dataset $\D$. As a first step we partition $\D$ into $k$ disjoint parts $\{\D_j\}_{j=1}^k$ each of size $N/k$. The gradient is the quantity
$$ g=\nabla_{\theta}L(\D;\theta)=\sum_{j=1}^k\nabla_{\theta}\ell(\D_j;\theta)=\sum_{j=1}^k g_j. $$
We refer to the terms $g_j\coloneqq\nabla_{\theta}\ell(\D_j;\theta)$ as \textit{partial gradients}.

In the distributed setting each worker node completes its job by returning a certain encoding of its assigned partial gradients. There can be different types of failures that may occur during the computation or the communication process. These failures are what we refer to as \textit{stragglers}, which are discarded by the main server. More specifically, the server only receives $f\coloneqq n-s$ completed tasks. Let $\I\subsetneq\N_n\coloneqq\{1,\cdots,n\}$ denote the set of indices of the $f$ \textit{fastest} workers who complete their tasks. Once \textit{any} set of $f$ tasks is received, the central server should be able to decode the received encoded partial gradients and recover the full gradient $g$.

\subsection{Gradient Coding}
\label{GC}

Gradient coding, proposed in \cite{TLDK17}, is a procedure comprised of an encoding matrix $\Bb\in\R^{n\times k}$, and a decoding vector $\ab_{\I}\in\R^n$; determined by $\I$. Schemes over $\C$ are also studied in \cite{HASH17,RTTD17}. Assuming that the workers have the same computational power, the same number of tasks is assigned to each of them. However, we deviate from this restriction in this paper in order to drop the assumption $(s+1)\mid n$. Each row of $\Bb$ corresponds to an encoding vector, also regarded as a task allocation, and each column corresponds to a data partition $\D_j$.

Each worker node is assigned a number of partial gradients from the partition, indexed by $\J_i\subsetneq\N_k$. The workers are tasked to compute an encoded version of the partial gradients $g_j\in\R^p$ corresponding to their assignments. Let
$$ \gb \coloneqq {\begin{pmatrix} | & | & & | \\ g_1 & g_2 & \hdots & g_k \\ | & | & & | \end{pmatrix}}^T \in \R^{k\times p} $$
denote the matrix whose rows constitute the transposes of the partial gradients, and the received encoded gradients are the rows of $\Bb_{\I}\gb$, for $\Bb_{\I}\in\{0,1\}^{f\times n}$ the submatrix of $\Bb$ consisting of the rows corresponding to $\I$. The full gradient of the objective (\ref{th_star_pr}) on $\D$ can be recovered by applying $\ab_{\I}$:
$$ g^T=\ab_{\I}^T(\Bb\gb)=\bold{1}_{1\times k}\gb= \sum_{j=1}^kg_j^{T},$$
provided that the encoding matrix $\Bb$ satisfies $\ab_{\I}^T\Bb=\bold{1}_{1\times k}$ for all ${{n}\choose{s}}$ possible index sets $\I$. Note that every partition is sent to $s+1$ servers, and each server will receive at least $\frac{k}{n}(s+1)$ distinct partitions. In sections \ref{bin_gc_sec} and \ref{dec_vector_sec}, we explain the design of our encoding matrix $\Bb$ and decoding vector $\ab_{\I}$, respectively. These may then be used for recovering the gradient $g$ at each iteration by the central server.

In \cite{TLDK17}, for a balanced  assignment, i.e., when all the workers are assigned the same number of tasks, the number of tasks corresponds to the support of the corresponding row of $\Bb$, and is lower bounded by $\|\Bb_{i*}\|_0\geq \frac{k}{n}(s+1)$. When this is met with equality for all rows of $\Bb$, the scheme is maximum distance separable (MDS). The restriction $(s+1)\mid n$ boils down to satisfying this bound, as $\frac{n}{s+1}$ needs to be an integer.

\section{Binary Gradient Coding --- Encoding Matrix}
\label{bin_gc_sec}

The main idea is to work with congruence classes $\bmod(s+1)$ on the set of the workers' indices $\N_n$, in such a way that the workers composing a congruence class are roughly assigned the same number of partitions (differing by no more than one), while all partitions appear exactly once in each class. By \textit{congruence class} we simply mean the set of integers $j\in\N_n$ which are equivalent $\bmod(s+1)$. The classes are denoted by $\left\{[i]_{s+1}\right\}_{i=0}^s$. One could use a random assignment once it has been decided how many partitions each worker is allocated. However, to get a deterministic encoding matrix, we assign the partitions in ``blocks'', i.e., submatrices consisting of only 1's. To simplify the presentation we will assume that $n=k$, though the idea can be easily adapted when $n\neq k$.

Define parameters $\ell$ and $r$ by performing Euclidean division, i.e., $n=\ell\cdot(s+1)+r$ such that $\ell=\floor{\frac{n}{s+1}}$ and $r=n-\ell\cdot(s+1)\equiv n\bmod(s+1)$. Similarly, for the integers $r,\ell$ we have $r=t\cdot\ell+q$. Therefore, $n=\ell\cdot(s+t+1)+q$. In a particular case, we will also need the parameters defined by the division of $n$ and $(\ell+1)$, which we define by $n=\lambda\cdot(\ell+1)+\rt$ (if $\ell=s-r$, then $\lambda=s$). To summarize, we have
\begin{itemize}[noitemsep,nolistsep]
  \item $n=\ell\cdot(s+1)+r \qquad \ 0\leq r<s+1$,
  \item $r=t\cdot\ell+q \qquad \qquad \ \ \ 0\leq q<\ell$,
  \item $n=\lambda\cdot(\ell+1)+\rt \qquad \ 0\leq \rt<\ell+1$,
\end{itemize}
where all terms are nonnegative integers.

In our proposed scheme, the encoding is identical for the classes $\Cfr_1\coloneqq\left\{[i]_{s+1}\right\}_{i=0}^{r-1}$, and is also identical for the classes $\Cfr_2\coloneqq\left\{[i]_{s+1}\right\}_{i=r}^{s}$. A more intuitive way of thinking about our design, is that we want $\Bb$ to be as close to a block diagonal matrix as possible. We refer to each disjoint set of consecutive $s+1$ rows of $\Bb$ as a \textit{block}, and the submatrix comprised of the last $r$ rows as the \textit{remainder block}. Note that in total we have $\ell+1$ blocks, and that each of the first $\ell$ blocks have workers with indices forming a complete residue system. We will present the two assignments (for $\Cfr_1$ and $\Cfr_2$) separately. Also, a certain numerical example, where $n=k=11$ and $s=3$ is presented for clarification.

\subsection{Repetition Assignment for Classes $0$ to $r-1 \ $ --- $\ \Cfr_1$}

In our construction each of the first $r$ residue classes also have an assigned row in the remainder block, such that we \textit{could} assign $r$ partitions to the last worker of each class in $\Cfr_1$, and evenly assign $s+1$ to all other workers corresponding to $\Cfr_1$. Our objective though is to distribute the $n$ tasks among the workers corresponding to the $\ell+1$ blocks as evenly as possible, for the congruence classes corresponding to $\Cfr_1$, in such a way that \textit{homogeneous} workers have similar loads. By homogeneous, we mean the workers have the same computational power, i.e., independent and identically distributed statistics for the computing time of similar tasks.

Note that $n=(\ell+1)\cdot s+(\ell+r-s)$, which implies that when $\ell>s-r$ we can assign $s+1$ tasks to each worker in the first $\ell+r-s$ blocks, and $s$ tasks to the remaining $s+1-r$ blocks. In the case where $\ell\leq s-r$, we assign $\lambda+1$ tasks to the first $\rt$ blocks and $\lambda$ tasks to the remaining $\ell+1-\rt$ blocks. It is worth pointing out that $\lambda=s$ and $\rt=0$ when $\ell=s-r$, which means that every worker corresponding to $\Cfr_1$ is assigned $\lambda=s$ tasks, as $n=(\ell+1)\cdot s$.

For example, for parameters $n=11$ and $s=3$ we get $\ell=2,r=3,t=1,q=1$, thus $\ell>r-s$; and the task allocation for $\Cfr_1$ is described by $\Bb_{\Cfr_1}\in\{0,1\}^{(\ell+1)\cdot r\times n}$:
\setcounter{MaxMatrixCols}{20}
$$ \Bb_{\Cfr_1} = \begin{bmatrix} 
\textbf{\blue 1} & \textbf{\blue 1} & \textbf{\blue 1} & \textbf{\blue 1} & & & & & & & \\
\textit{\cyan 1} & \textit{\cyan 1} & \textit{\cyan 1} & \textit{\cyan 1} & & & & & & & \\
\mathfrak{\purple 1} & \mathfrak{\purple 1} & \mathfrak{\purple 1} & \mathfrak{\purple 1} & & & & & & & \\
& & & & \textbf{\blue 1} & \textbf{\blue 1} & \textbf{\blue 1} & \textbf{\blue 1} & & & \\
& & & & \textit{\cyan 1} & \textit{\cyan 1} & \textit{\cyan 1} & \textit{\cyan 1} & & & \\
& & & & \mathfrak{\purple 1} & \mathfrak{\purple 1} & \mathfrak{\purple 1} & \mathfrak{\purple 1} & & & \\
& & & & & & & & \textbf{\blue 1} & \textbf{\blue 1} & \textbf{\blue 1} \\
& & & & & & & & \textit{\cyan 1} & \textit{\cyan 1} & \textit{\cyan 1} \\
& & & & & & & & \mathfrak{\purple 1} & \mathfrak{\purple 1} & \mathfrak{\purple 1}
\end{bmatrix}, $$
where each congruence class is represented by a different color and font. The indicated dimensions are for the case where $r>0$, i.e., the remainder block is not empty. An explicit implementation is described (using matlab notation) in algorithm \ref{B_ones_unb_unif_C1}, where $\tilde{\Bb}_{\Cfr_1}$ is obtained from $\Bb_{\Cfr_1}$ by properly appending zero vectors. For coherence, we index the rows by $i$ starting from 0, and the columns by $j$ starting from $1$.

\begin{algorithm}[h]
\label{B_ones_unb_unif_C1}
\SetAlgoLined
{\footnotesize
  \KwIn{number of workers $n$ and stragglers $s$, where $s,n\in\Z_+$}
  \KwOut{encoding matrix $\tilde{\Bb}_{\Cfr_1}\in\{0,1\}^{n\times n}$ 
  \Comment{assume $n=k$}}
  $\tilde{\Bb}_{\Cfr_1}\gets\bold{0}_{n\times n}$, and use the division algorithm to get the parameters: \\
  $\ind n=\ell\cdot(s+1) \qquad r=t\cdot\ell+q \qquad n=\lambda\cdot(s+1)+\rt$ \\

  \For{$i=0$ to $r-1$}
  {
    \If{$\ell+r>s$}
    {
      \For{$j=1$ to $\ell+r-s$}
      {
        $\tilde{\Bb}_{\Cfr_1}\Big[(j-1)(s+1)+i,(j-1)(s+1)+1:j(s+1)\Big]=\bold{1}_{1\times(s+1)}$ \\
      }
      \For{$j=\ell+r-s+1$ to $\ell+1$}
      {
        $\tilde{\Bb}_{\Cfr_1}\Big[(j-1)(s+1)+i,(j-1)s+(\ell+r-s)+1:(j-1)s+\ell+r\Big]=\bold{1}_{1\times s}$
      }
    }
    \ElseIf{$\ell+r\leq s$}
    {
      \For{$j=1$ to $\rt$}
      {
        $\tilde{\Bb}_{\Cfr_1}\Big[(j-1)(s+1)+i,(j-1)(\lambda+1)+1:j(\lambda+1)\Big]=\bold{1}_{1\times(\lambda+1)}$ \\
      }
      \For{$j=\rt+1$ to $\ell+1$}
      {
        $\tilde{\Bb}_{\Cfr_1}\Big[(j-1)(s+1)+i,(j-1)\lambda+\rt+1:(j-1)\lambda+\rt+\lambda\Big]=\bold{1}_{1\times\lambda}$
      }
    }
  }
 \Return $\tilde{\Bb}_{\Cfr_1}$
 \caption{Determining $\tilde{\Bb}_{\Cfr_1}$ --- $\Cfr_1=\left\{[i]_{s+1}\right\}_{i=0}^{r-1}$}
}
\end{algorithm}

\subsection{Repetition Assignment for Classes $r$ to $s \ $  --- $\ \Cfr_2$}

For the workers corresponding to $\Cfr_2$, we first check if $q=0$. If this is the case, we distribute evenly the $n$ partitions between the workers to each $i\in\Cfr_2$, i.e., each worker is assigned $(s+t+1)$ partitions; as $n=\ell\cdot(s+t+1)$ and here we are only considering $\ell$ blocks. When $0<q<r$, we assign $(s+t+2)$ tasks to each worker of $\Cfr_2$ in the first $q$ blocks, and $(s+t+1)$ to the workers in the remaining $\ell-q$ blocks.

In the numerical example considered, we have $q=1$ and $\Bb_{\Cfr_2}\in\{0,1\}^{\ell\cdot(s+1-r)\times n}$:
$$ \Bb_{\Cfr_2} = \begin{bmatrix} 
{\dgreen 1} & {\dgreen 1} & {\dgreen 1} & {\dgreen 1} & {\dgreen 1} & {\dgreen 1} & & & & & \\
& & & & & & {\dgreen 1} & {\dgreen 1} & {\dgreen 1} & {\dgreen 1} & {\dgreen 1}
\end{bmatrix}. $$
An explicit implementation is provided in algorithm \ref{B_ones_unb_unif_C2}, where $\tilde{\Bb}_{\Cfr_2}$ is obtained from $\Bb_{\Cfr_2}$ by properly appending zero vectors.

\begin{algorithm}[h]
\label{B_ones_unb_unif_C2}
\SetAlgoLined
{\footnotesize
  \KwIn{number of workers $n$ and stragglers $s$, where $s,n\in\Z_+$}
  \KwOut{encoding matrix $\tilde{\Bb}_{\Cfr_2}\in\{0,1\}^{n\times n}$ 
  \Comment{assume $n=k$}}
  $\tilde{\Bb}_{\Cfr_2}\gets\bold{0}_{n\times n}$, and use the division algorithm to get the parameters: \\
  $\ind n=\ell\cdot(s+1) \qquad r=t\cdot\ell+q \qquad n=\lambda\cdot(s+1)+\rt$ \\

  \For{$i=r$ to $r$}
  {
    \If{$q=0$}
    {
      \For{$j=1$ to $\ell$}
      {
        $\tilde{\Bb}_{\Cfr_2}\Big[(j-1)(s+1)+i,(j-1)(s+t+1)+1:j(s+t+1)\Big]=\bold{1}_{1\times(s+t+1)}$ \\
      }
    }
    \ElseIf{$q>0$}
    {
      \For{$j=1$ to $q$}
      {
        $\tilde{\Bb}_{\Cfr_2}\Big[(j-1)(s+1)+i,(j-1)(s+t+2)+1:j(s+t+1)\Big]=\bold{1}_{1\times(s+t+2)}$ \\
      }
      \For{$j=q+1$ to $\ell$}
      {
        $\Bb\Big[(j-1)(s+1)+i,(j-1)(s+t+1)+q+1:j(s+t+1)+q\Big]=\bold{1}_{1\times(s+t+1)}$
      }
    }
  }
 \Return $\tilde{\Bb}_{\Cfr_2}$
 \caption{Determining $\tilde{\Bb}_{\Cfr_2}$ --- $\Cfr_2=\left\{[i]_{s+1}\right\}_{i=r}^{s}$}
}
\end{algorithm}

The final step is to \textit{combine} the two matrices to get $\Bb$. One could merge the two algorithms into one, or run them separately to get $\Bb=\tilde{\Bb}_{\Cfr_1}+\tilde{\Bb}_{\Cfr_2}$, demonstrated as follows:
$$ \Bb = \begin{bmatrix} 
\textbf{\blue 1} & \textbf{\blue 1} & \textbf{\blue 1} & \textbf{\blue 1} & & & & & & & \\
\textit{\cyan 1} & \textit{\cyan 1} & \textit{\cyan 1} & \textit{\cyan 1} & & & & & & & \\
\mathfrak{\purple 1} & \mathfrak{\purple 1} & \mathfrak{\purple 1} & \mathfrak{\purple 1} & & & & & & & \\
{\dgreen 1} & {\dgreen 1} & {\dgreen 1} & {\dgreen 1} & {\dgreen 1} & {\dgreen 1} & & & & & \\
& & & & \textbf{\blue 1} & \textbf{\blue 1} & \textbf{\blue 1} & \textbf{\blue 1} & & & \\
& & & & \textit{\cyan 1} & \textit{\cyan 1} & \textit{\cyan 1} & \textit{\cyan 1} & & & \\
& & & & \mathfrak{\purple 1} & \mathfrak{\purple 1} & \mathfrak{\purple 1} & \mathfrak{\purple 1} & & & \\
& & & & & & {\dgreen 1} & {\dgreen 1} & {\dgreen 1} & {\dgreen 1} & {\dgreen 1} \\
& & & & & & & & \textbf{\blue 1} & \textbf{\blue 1} & \textbf{\blue 1} \\
& & & & & & & & \textit{\cyan 1} & \textit{\cyan 1} & \textit{\cyan 1} \\
& & & & & & & & \mathfrak{\purple 1} & \mathfrak{\purple 1} & \mathfrak{\purple 1}
\end{bmatrix} \in \{0,1\}^{n\times n} $$

The encoding matrix $\Bb$ is also the adjacency matrix of a bipartite graph $G=(\mathcal{L},\mathcal{R},\mathcal{E})$, where the vertices $\mathcal{L}$ and $\mathcal{R}$ correspond to the $n$ workers and the $k$ partitions, respectively. We can also vary the number of stragglers $s$ the scheme can tolerate for a fixed $n$, by trading the sparsity of $\Bb$. In other words, if $\Bb$ is designed to tolerate more stragglers, then more overall partial gradients need to be computed; as $\|\Bb\|_F^2=|\text{supp}(\Bb)|=k\cdot(s+1)$.

\section{Binary Gradient Coding --- Decoding Vector}
\label{dec_vector_sec}

Another drawback of the binary scheme in \cite{TLDK17} is that it computes and stores all scenarios for decoding vectors in a matrix $\Ab\in\R^{{{n}\choose{f}}\times n}$, where a matrix inversion is required to compute $\Ab$. This matrix needs to be stored and searched through at each iteration of the gradient descent procedure. We propose a more efficient online decoding algorithm.

In any straggler scenario, since there is no rescaling of the partial gradients taking place by encoding with $\Bb$ as the coefficients are $1$ or $0$, decoding is nothing but summing a certain subset of the received encoded tasks, while making sure that no partial gradient is added more than once. Hence, among any $f$ workers who send back their computed sum of partial gradients, there should be $\ell$ workers for $\ell\coloneqq\frac{n}{s+1}\in\Z_+$ (or $\ell+1$ where $\ell=\floor{\frac{n}{s+1}}$, if $(s+1)\nmid n$) which have no common assigned partitions. This is shown next.

If $r=0$, for the decoding vector $\ab_{\I}$ we traverse through the $s+1$ classes in order, to detect one which is a complete residue system (algorithm \ref{determine_aI}). In the case where $r>0$, we first traverse through the last $s+1-r$ congruence classes; checking only the first $\ell$ blocks. If there is not a complete residue system from the received workers, we proceed to the first $r$ classes; checking also the remainder block. This extra step is to make the scheme more efficient. In both cases, by the pigeonhole principle we are guaranteed to have a complete residue system when $f$ tasks are received.

The next step is to devise a decoding vector for each of the ${n}\choose{s}$ different straggler scenarios $\I$. We associate each complete residue system $\ell$-tuple (or $(\ell+1)$)
with a decoding vector $\ab_i$ 
$$ \ab_i\coloneqq \sum_{j\in[i]_{\ell}}\eb_j \ \in\{0,1\}^{1\times n}, $$
for $i\in\N_{\ell}-1$, where $\eb_j$ denotes the $j^{th}$ standard basis vector of $\R^n$. Also, note that $\|\ab_i\|_{0}=\ell+1$ for the decoding vectors corresponding to the first $r$ classes, and $\|\ab_i\|_{0}=\ell$ for the remaining classes. In both cases, $\ab_{i+1}$ is a cyclic shift of $\ab_i$.

At each iteration the gradient is computed once $f$ worker tasks are received. Define the \textit{received-indicator vectors}
$$ \left(\text{rec}_{\I}\right)_i = \begin{cases} 1 \qquad \text{ if } i\in\I \\ 0 \qquad \text{ if } i\not\in\I \end{cases}, $$
for each possible $\I$, where $\|\text{rec}_{\I}\|_0=f$ and $n-\|\text{rec}_{\I}\|_0=s$. Thus, there is at least one $i\in\N_{\ell}-1$ for which $\text{supp}(\ab_i)\subsetneq \text{supp}(\text{rec}_{\I})$. If there are multiple $\ab_i$ satisfying this property, any one can be selected. The pseudocode is presented in algorithm \ref{determine_aI}.

\begin{algorithm}[h]
\label{determine_aI}
\SetAlgoLined
\KwIn{received indicator-vector $\text{rec}_{\I}$}
\KwOut{decoding vector $\ab_{\I}$}
\If{r=0}
{
 \For{$i = 0$ to $s$}
  {
   \If {$\left(\mathrm{rec}_{\I}\right)_i = 1$}
    {
     $l\gets i$ \\
     \If {$\mathrm{supp}(\ab_l)\subseteq \mathrm{supp}(\mathrm{rec}_{\I})$}
     {
      $\ab\gets \ab_l$ \\
      break
     }
    }
  }
}
\ElseIf{$r>0$}
{run the above for-loop \textbf{for} $i=r \textit{ to } s$\\ and then \textbf{for} $i=0 \textit{ to } r-1$}
 \Return $\ab_{\I}\gets\ab$
 \caption{Determining $\ab_{\I}$}
\end{algorithm}

\begin{Thm}
The gradient coding scheme comprised of $\Bb$ and $\ab_{\I}$ based on algorithms \ref{B_ones_unb_unif_C1},\ref{B_ones_unb_unif_C2},\ref{determine_aI}, is robust to $s$ stragglers.
\end{Thm}

\begin{proof}
By the pigeonhole principle we require
\begin{align*}
  \nu &\coloneqq \ell\cdot r+(\ell-1)\cdot\big[(s+1)-r\big]+1\\
  &= \ell\cdot(s+1)-s+r = \big[\ell\cdot(s+1)+r\big]-s=n-s
\end{align*}
workers to send their task, which implies the scheme is robust to $s$ stragglers.
\end{proof}

Note that the total number of task assignments is $k\cdot (s+1)$, for any pair $(s,n)$ where $s<n$, as expected. This is the same total load required in the MDS based schemes, which directly relates to the bound mentioned in \ref{GC}.

We point out that a similar decoding appears in \cite{CP18}, which deals with approximating the gradient. The runtime of algorithm \ref{determine_aI} is $O((\ell+1)\cdot(s+1))$. With the modification of breaking out of the for-loop early by only traversing through the classes $0,\cdots,s-1$, and assigning $\ab_{\I}\gets\ab_{s}$ if none was selected, the runtime is reduced to $O((\ell+1)\cdot s)$. This is significantly faster than the decoding algorithm of \cite{TLDK17}, as the size of the corresponding decoding matrix $\Ab$ grows exponentially in $n$ (unless $s$ is fixed, for which ${n}\choose{s}$ is polynomial in $n$). 

\section{Close to Uniform Assignment Distribution}
\label{cl_un_ass_distr_sec}

A drawback of the proposed scheme is that the load assignments can have a wide range depending on how small $r$ is compared to $s+1$. This is due to the lighter load assigned to workers in the remainder block of size $r$. This can be regarded as the cost for dropping the assumption $(s+1)\mid n$, which does not often hold for a pair of two random positive integers --- for fixed $n$ and random $s<n$; it holds with probability $\frac{\sigma_0(n)-2}{n}$, where $\sigma_0$ the divisor function of the $0^{th}$ power.


By our construction, we assign tasks to the servers such that the difference in the load assigned to any two workers within the same set of classes $\Cfr_1$ or $\Cfr_2$ is at most one. This is evident in algorithms \ref{B_ones_unb_unif_C1} and \ref{B_ones_unb_unif_C2}.

\begin{Def}
\label{def_unif}
Define $d_s(\Bb)\coloneqq\sum\limits_{i=1}^n\left|\|\Bb_{i*}\|_0-\frac{k}{n}(s+1)\right|$ for $\Bb\in\N_0^{n\times k}$, where $\N_0\coloneqq\{0,1,2,\cdots\}$. This distance measures how far the task allocations $\left\{\|\Bb_{i*}\|_0\right\}_{i=1}^n$ are from being \textbf{uniform}, i.e., $\|\Bb_{i*}\|_0= \floor{\frac{k}{n}(s+1)+\frac{1}{2}}$ for all $i\in\N_n$. Furthermore, $\{\|\Bb_{i*}\|_0\}_{i=1}^n$ is uniform; i.e., all elements are equal, if and only if $d_s(\Bb)=0$.
\end{Def}

\begin{Prop}
\label{prop_optimality}
The task allocation through $\Bb$ resulting from algorithms \ref{B_ones_unb_unif_C1} and \ref{B_ones_unb_unif_C2} is a solution to the optimization problem:
$$ \mathrm{(IP)} \qquad \min_{\Bb\in\N_0^{n\times k}}\big\{d_s(\Bb)\big\}, $$
such that $\sum\limits_{i=1}^n\|\Bb_{i*}\|_0=k\cdot(s+1)$.
\end{Prop}

Proposition \ref{prop_optimality} holds for permutations of the columns of $\Bb$, or a random assignment of tasks per class; as opposed to repeating blocks --- as long as all partitions are present only once in a single worker of each congruence class. The decoding in either of these cases remains the same. Moreover, the solution to (IP) is unique up to a permutation of rows and columns of $\Bb$.

\subsection{Distribution of Assignments for $n\geq s^2$}
\label{n_gr_s2_subsec}

Considering the identities from Section \ref{bin_gc_sec}, note that for $\ell>r$ we have $t=0$ and $r=q$. Furthermore, when $\ell=s$ we have $n=s\cdot(s+1)+r\simeq s^2$, and in the regime $n\geq s^2$, there are only three values for $n$ for which $t=1$. Consequently, for $\ell+r>s$, the difference in the number of allocated partitions to each worker will not exceed 3.

\begin{Lemma}
Let $n=s^2+a$ for $a\in \N_0$ and $s<n$. Then, we have $t=1$ only when $a=s-2,s-1$ or $2s$. Otherwise, $t=0$.
\end{Lemma}

\begin{proof}
We break up the proof into three cases:

\noindent \underline{Case $a\in\{0,\cdots,s-3\}$}: For $\alpha = s-a\in\{3,4,\cdots,s\}$:
$$ n = s\cdot(s+1)-\alpha = \overbrace{(s-1)}^{\ell}\cdot(s+1)+\overbrace{(s+1-\alpha)}^r, $$
and $\ell>r$ for any $\alpha$. Thus, $t=0$ and $r=q$.

\noindent \underline{Case $a\in\{s,\cdots,2s-1\}$}: Let $n=s^2+a=s^2+(s+\beta)$ for $\beta\in\{0,\cdots,s-1\}$. Then $n=s\cdot(s+1)+\beta$ implies $\ell=s$ and $r=\beta$. Since $r<\ell$, it follows that $t=0$ and $r=q$.

\noindent \underline{Case $a\gneq 2s$}: The final case to consider is $a=2s+\gamma$, for $\gamma\in\Z_+$. The resulting parameters are $r=q=\text{rem}\big(\text{rem}\big(\gamma,s+1\big)-1,s+1\big)$, $\ell=(s^2+2s+\gamma-r)/(s+1)$ and $t=0$, where rem$(\cdot,\cdot)$ is the remainder function.

When $\alpha=1$ it follows that $r=s$ and $\ell=s-1$, thus $t=1$ and $q=1$. For $\alpha=2$ we get $r=\ell=s-1$, hence $t=1$ and $q=0$. For both $\alpha=1$ and $\alpha=2$; $t=1$ is a consequence of $r\geq \ell$. In addition, when $\beta=s$ we have $r=\ell=s$; thus $t=1$ and $q=0$.

\end{proof}

\section{Allocation for Heterogeneous Workers}
\label{heter_case_sec}

For homogeneous workers, we allocated the partitions as uniform as possible, according to Definition \ref{def_unif}. In this section we discuss how to allocate the partitions when the workers are of \textit{heterogeneous} nature, i.e., when they have different computational power. This should be done in such a way that all workers have the same expected execution time; as the stragglers are assumed to be uniformly random. We present the case where we have two groups of machines, each consisting of the same type. The analysis for more groups is analogous. Similar ideas appear in \cite{TLDK17, OGU19}, in different contexts.

The two types of workers are denoted by $\T_i$; with a total of $\tau_i$ machines, and their expected execution for computing $g_j$ (for equipotent $\D_j$'s) by
$$ t_i\coloneqq\E\left[\text{time for } \T_i \text{ to compute } g_j\right], $$
for $i\in\{1,2\}$, where $t_1\lneq t_2$; i.e., machines $\T_1$ are faster. Let $|\J_{\T_i}|$ denote the number of partitions each worker of $\T_i$ is assigned. The goal is to find $|\J_{\T_1}|$ and $|\J_{\T_2}|$ so that
\begin{equation}
\label{expecation_condition}
  \E\left[\T_1 \text{ compute their task}\right]=\E\left[\T_2 \text{ compute their task}\right],
\end{equation}
implying $t_1\cdot|\J_{\T_1}|=t_2\cdot|\J_{\T_2}|$. Hence $|\J_{\T_1}|\gneq|\J_{\T_2}|$, as $t_1\lneq t_2$. Let $\tau_1=\frac{\alpha}{\beta}\cdot\tau_2$ for $\frac{\alpha}{\beta}\in\Q_+$ in reduced form. Since $\tau_1+\tau_2=n$, it follows that
$$ \tau_1=\frac{\alpha}{\alpha+\beta} n \qquad \text{ and } \qquad \tau_2=\frac{\beta}{\alpha}\tau_1=\frac{\beta}{\alpha+\beta} n. $$
To simplify the presentation of the task assignments, we assume $(s+1)\mid n$. If $(s+1)\nmid n$, one can follow a similar approach to that presented in section \ref{bin_gc_sec} to obtain a close to uniform task allocation; while \textit{approximately} satisfying \eqref{expecation_condition}.

The main idea is to fully partition the data across the workers, such that each congruence class is comprised of roughly $\frac{\alpha}{\alpha+\beta}\cdot\frac{k}{s+1}$ workers of type $\T_1$, and $\frac{\beta}{\alpha+\beta}\cdot\frac{k}{s+1}$ workers of type $\T_2$. We want $\frac{\tau_1+\tau_2}{s+1}=\frac{n}{s+1}$ many workers for each congruence class, and
$$ |\J_{\T_1}|\cdot\frac{\tau_1}{s+1}+|\J_{\T_2}|\cdot\frac{\tau_2}{s+1}=k $$
partitions to be assigned to each class. That is, the dataset $\D$ is completely distributed across each congruence class, and our gradient coding scheme is designed accordingly.

Putting everything together, the following  conditions determine $|\J_{\T_1}|$ and $|\J_{\T_2}|$
\begin{enumerate}[label=(\roman*)]
  \item $t_1\cdot|\J_{\T_1}|=t_2\cdot|\J_{\T_2}| \quad \iff \quad |\J_{\T_2}| = \frac{t_1}{t_2}\cdot|\J_{\T_2}|$
  \item $|\J_{\T_1}|\cdot\tau_1+|\J_{\T_2}|\cdot\tau_2=(s+1)\cdot k$
  \item $\tau_2=\frac{\beta}{\alpha}\cdot\tau_1 \quad \iff \quad \tau_1=\frac{\alpha}{\beta}\cdot\tau_2$.
\end{enumerate}
By substituting (iii) into (ii) to solve for $|\J_{\T_2}|$, and then plugging it into (i) to solve for $|\J_{\T_1}|$, we get
\begin{itemize}
  \item $|\J_{\T_1}| = (s+1)\cdot k \cdot \left(\frac{\alpha t_2}{\alpha t_2+\beta t_1}\right)\cdot\frac{1}{\tau_1}$
  \item $|\J_{\T_2}| = (s+1)\cdot k \cdot \left(\frac{\beta t_1}{\alpha t_2+\beta t_1}\right)\cdot\frac{1}{\tau_2}$.
\end{itemize}
which we round to get appropriate numbers of assignments.

This framework may be generalized to any number of different groups of machines. Under the same assumptions, for $\T_1,\cdots,\T_m$ different groups with $t_i\lneq t_{i+1}$ for all $i\in\N_{m-1}$:
\begin{enumerate}[label=(\roman*)]
  \item $t_1\cdot|\J_{\T_1}|=t_2\cdot|\J_{\T_2}|=\cdots=t_m\cdot|\J_{\T_m}|$
  \item $|\J_{\T_1}|\cdot\tau_1+|+|\J_{\T_2}|\cdot\tau_2+\cdots+|\J_{\T_m}|\cdot\tau_m=(s+1)\cdot k$
  \item $\tau_1=\frac{\alpha_2}{\beta_2}\cdot\tau_2=\cdots=\frac{\alpha_m}{\beta_m}\cdot\tau_m$, for $\frac{\alpha_{i+1}}{\beta_{i+1}}\in\Q_+$
\end{enumerate}
need to be met. This gives us a system of $2(m-1)+1=2m-1$ equations with $m$ unknowns $\{|\J_{\T_j}|\}_{j=1}^m$, which is solvable.

\section{Acknowledgement}
This work was partially supported by grant ARO W911NF-15-1-0479.


\bibliographystyle{IEEEtran}
\bibliography{refs_all}

\end{document}